\documentclass[journal,onecolumn,12pt]{IEEEtran}
\usepackage{pslatex} 
\usepackage{amsfonts,color,morefloats}
\usepackage{amssymb,amsmath,latexsym,amsthm}

\newcommand\myatop[2]{\genfrac{}{}{0pt}{}{#1}{#2}}

\newcommand{\m}{{\mathrm{m}}}

\newcommand{\wt}{{\mathrm{wt}}} 

\newcommand{\lcm}{{\rm lcm}}

\newcommand{\ord}{{\mathrm{ord}}}

\newcommand{\cPGRM}{{\mathtt{PGRM}}} 
\newcommand{\cGRM}{{{\mho}}} 

\newcommand{\gf}{{\mathrm{GF}}}

\newcommand{\C}{{\mathcal{C}}}
\newcommand{\BCH}{{\mathtt{BCH}}}

\newcommand{\cR}{{\mathcal{R}}}
\newcommand{\cJ}{{\mathtt{J}}}
\newcommand{\cN}{{\mathcal{N}}}
\newcommand{\AGL}{{\mathrm{AGL}}}

\newcommand{\Aut}{{\mathrm{Aut}}} 
\newcommand{\PAut}{{\mathrm{PAut}}} 
\newcommand{\MAut}{{\mathrm{MAut}}} 
\newcommand{\GAut}{{\mathrm{Aut}}}

\newcommand{\Sym}{{\mathrm{Sym}}} 

\newcommand{\cP}{{\mathcal{P}}}
\newcommand{\cB}{{\mathcal{B}}} 
 
\newcommand{\cD}{{\mathcal{D}}} 
\newcommand{\bD}{{\mathcal{D}}}

\newcommand{\bone}{{\mathbf{1}}}

\newtheorem{theorem}{Theorem}
\newtheorem{lemma}[theorem]{Lemma}

\newtheorem{corollary}[theorem]{Corollary}

\newtheorem{problem}{Open Problem}

\newtheorem{example}{Example}

\begin{document}

\title{Another Generalisation of the Reed-Muller Codes\thanks{C. Ding's research was supported by the Hong Kong Research Grants Council, under Grant No. 16301114.}}

\author{Cunsheng~Ding, ~Chunlei Li, and Yongbo Xia   
\thanks{C. Ding is with the Department of Computer Science and Engineering, 
The Hong Kong University of Science and Technology, Clear Water Bay, Kowloon, Hong Kong, China (Email: cding@ust.hk).} 
\thanks{C. Li is with the Department of Electrical Engineering and Computer Science, University of Stavanger, Stavanger, 4036, Norway (Email: chunlei.li@uis.no).}
\thanks{Y. Xia is with the Department of Mathematics and Statistics, South-Central University for Nationalities, Wuhan 430074, China (Email: xia@mail.scuec.edu.cn).}
}

\date{\today}
\maketitle

\begin{abstract} 
The punctured binary Reed-Muller code is cyclic and was generalised into the punctured generalised 
Reed-Muller code over $\gf(q)$ in the literature. The major objective of this paper is to present another generalisation of the punctured binary Reed-Muller code and the binary Reed-Muller code, and analyse these codes. 
We will prove that this newly generalised Reed-Muller code is affine-invariant and holds $2$-designs. 
Another objective is to construct 
a family of reversible cyclic codes that are related to the newly generalised Reed-Muller codes.  
\end{abstract}

\begin{keywords}
BCH codes, cyclic codes, linear codes, punctured Reed-Muller codes, punctured generalised Reed-Muller 
codes, Reed-Muller codes. 
\end{keywords}

\section{Introduction}\label{sec-intro} 

Throughout this paper, let $p$ be a prime and let $q=p^s$ be a power of $p$, where $s$ is a positive integer. 
An $[n,k,d]$ code $\C$ over $\gf(q)$ is a $k$-dimensional subspace of $\gf(q)^n$ with minimum 
(Hamming) distance $d$. 

An $[n,k]$ code $\C$ over $\gf(q)$ is called {\em cyclic} if 
$(c_0,c_1, \cdots, c_{n-1}) \in \C$ implies $(c_{n-1}, c_0, c_1, \cdots, c_{n-2}) 
\in \C$.  
By identifying any vector $(c_0,c_1, \cdots, c_{n-1}) \in \gf(q)^n$ 
with  
$$ 
c_0+c_1x+c_2x^2+ \cdots + c_{n-1}x^{n-1} \in \gf(q)[x]/(x^n-1), 
$$
any code $\C$ of length $n$ over $\gf(q)$ corresponds to a subset of the quotient ring 
$\gf(q)[x]/(x^n-1)$. 
A linear code $\C$ is cyclic if and only if the corresponding subset in $\gf(q)[x]/(x^n-1)$ 
is an ideal of the ring $\gf(q)[x]/(x^n-1)$. 

Note that every ideal of $\gf(q)[x]/(x^n-1)$ is principal. Let $\C=\langle g(x) \rangle$ be a 
cyclic code, where $g(x)$ is monic and has the smallest degree among all the 
generators of $\C$. Then $g(x)$ is unique and called the {\em generator polynomial,} 
and $h(x)=(x^n-1)/g(x)$ is referred to as the {\em check polynomial} of $\C$. 

The original Reed-Muller codes were discovered by Reed and Muller independently in 1964 
\cite{Muller,Reed}. These codes are standard materials in textbooks and research monographs 
on coding theory, and were employed in space communication in the Mariner 9 Spacecraft. 
These facts show the importance of the Reed-Muller codes. 

Recently, Reed-Muller codes have become a hot topic in coding theory due to the fact that 
they belong to the classes of locally testable codes and locally decodable codes, which  
makes them useful in the design of probabilistically checkable proofs in computational 
complexity theory \cite{Sergey}. 

The original Reed-Muller codes are binary and linear, but not cyclic. However, the punctured 
Reed-Muller codes are cyclic. The punctured binary Reed-Muller code was generalised into the 
punctured generalised 
Reed-Muller code over $\gf(q)$ in the literature \cite{AssmusKey92,AssmusKey98,DK00,DGM,GV98}. In this paper, we present another 
generalisation of the punctured binary Reed-Muller code and the binary Reed-Muller code, and study properties of these  
codes. We will prove that the newly generalised Reed-Muller code is affine-invariant and holds $2$-designs. 
We will also construct a family of reversible cyclic codes from the newly generalised Reed-Muller 
codes and analyse their parameters.

\section{$q$-cyclotomic cosets modulo $n$ and auxililaries}\label{sec-qcyclotomiccosets} 

To deal with cyclic codes of length $n$ over $\gf(q)$, we have to study the canonical factorization of $x^n-1$ 
over $\gf(q)$. To this end, we need to introduce $q$-cyclotomic cosets modulo $n$. Note that $x^n-1$ has no 
repeated factors over $\gf(q)$ if and only if $\gcd(n, q)=1$. Throughout this paper, we assume 
that $\gcd(n, q)=1$.  

Let $\cN = \{0,1,2, \cdots, n-1\}$, denoting the ring of integers modulo $n$. For any $s \in \cN$, the \emph{$q$-cyclotomic coset of $s$ modulo $n$\index{$q$-cyclotomic coset modulo $n$}} is defined by 
$$ 
C_s=\{sq^i \bmod{n}: 0 \leq i \leq \ell_s-1\} \subseteq \cN,  
$$
where $\ell_s$ is the smallest positive integer such that $s \equiv s q^{\ell_s} \pmod{n}$, and is the size of the 
$q$-cyclotomic coset. The smallest integer in $C_s$ is called the \emph{coset leader\index{coset leader}} of $C_s$. 
Let $\Gamma_{(n,q)}$ be the set of all the coset leaders. We have then $C_s \cap C_t = \emptyset$ for any two 
distinct elements $s$ and $t$ in  $\Gamma_{(n,q)}$, and  
\begin{eqnarray}\label{eqn-cosetPP}
\bigcup_{s \in  \Gamma_{(n,q)} } C_s = \cN. 
\end{eqnarray}
Hence, the distinct $q$-cyclotomic cosets modulo $n$ partition $\cN$. 

Let $m=\ord_{n}(q)$, and let $\alpha$ be a generator of $\gf(q^m)^*$. Put $\beta=\alpha^{(q^m-1)/n}$. 
Then $\beta$ is a primitive $n$-th root of unity in $\gf(q^m)$. The minimal 
polynomial $\m_{s}(x)$ of $\beta^s$ over $\gf(q)$ is a monic polynomial of the smallest degree over 
$\gf(q)$ with 
$\beta^s$ as a zero. It is now straightforward to see that this polynomial is given by 
\begin{eqnarray}
\m_{s}(x)=\prod_{i \in C_s} (x-\beta^i) \in \gf(q)[x], 
\end{eqnarray} 
which is irreducible over $\gf(q)$. It then follows from (\ref{eqn-cosetPP}) that 
\begin{eqnarray}\label{eqn-canonicalfact}
x^n-1=\prod_{s \in  \Gamma_{(n,q)}} \m_{s}(x)
\end{eqnarray}
which is the factorization of $x^n-1$ into irreducible factors over $\gf(q)$. This canonical factorization of $x^n-1$ 
over $\gf(q)$ is fundamental for the study of cyclic codes.

%The following result will be useful and is not hard to prove \cite{HP03}[Theorem 4.1.4] 

%\begin{lemma}
%The size $\ell_s$ of each $q$-cyclotomic coset $\C_s$ is a divisor of $\ord_{n}(q)$, which is the size %$\ell_1$ of $C_1$. 
%\end{lemma}

%The following lemma and theorem were proved in \cite{AKS} and contain results in \cite{YH96} as special %cases. 

%\begin{lemma}\label{lem-AKS}
%Let $n$ be a positive integer such that $\gcd(n, q)=1$ and $q^{\lfloor m/2 \rfloor}<n \leq q^m-1$, %where 
%$m=\ord_n(q)$. Then the $q$-cyclotomic coset $C_s=\{sq^j \bmod{n}: 0 \leq j \leq m-1\}$ has cardinality 
%$m$ for all $s$ in the range $1 \leq s \leq n q^{\lceil m/2 \rceil}/(q^m-1)$. In addition, every $s$ %with 
%$s \not\equiv 0 \pmod{q}$ in this range is a coset leader.  
%\end{lemma} 

\section{The punctured generalised Reed-Muller codes over $\gf(q)$}

Let $m$ be a positive integer and let $n=q^m-1$ from now on, where $q=p^s$, $p$ is a prime 
and $s$ 
is a positive integer. For any integer $a$ with $0 \leq a \leq n-1$, we have the following $q$-adic 
expansion 
$$ 
a=\sum_{j=0}^{m-1} a_jq^j, 
$$
where $0 \leq a_j \leq q-1$. The $q$-weight of $a$, denoted by $\varpi(a)$, 
is defined by 
$$
\varpi(a)=\sum_{j=0}^{m-1} a_j. 
$$

Let $\alpha$ be a generator of $\gf(q^m)^*$. Let $\ell=\ell_1(q-1)+\ell_0<q(m-1),$ 
where $0 \leq \ell_0 \leq q-1$. The $\ell$-th order 
\emph{punctured generalised Reed-Muller code}\index{punctured generalised Reed-Muller code} 
$\cPGRM_q(\ell, m)$ over $\gf(q)$ is the cyclic code of length $n=q^m-1$ with generator polynomial 
\begin{eqnarray}\label{eqn-generatorpolyPGRMcode}
g_{(q,m,\ell)}^{R}(x):= \prod_{\myatop{1 \leq a \leq n-1}{ \varpi(a) < (q-1)m-\ell}} (x - \alpha^a).  
\end{eqnarray}
Since $\varpi(a)$ is a constant function 
on each $q$-cyclotomic coset modulo $n$, $g_{(q,m,\ell)}^{R}(x)$ is a polynomial over 
$\gf(q)$. By definition, $g_{(q,m,\ell)}^{R}(x)$ is a divisor of $x^n-1$. 

The parameters of the punctured generalised Reed-Muller code are known and given in the following theorem 
\cite[Theorem 5.4.1]{AssmusKey92}. 

\begin{theorem}\label{thm-PGRMcode}
The code $\cPGRM_q(\ell, m)$ has length $n=q^m-1$, dimension 
$$ 
k=\sum_{i=0}^\ell \sum_{j=0}^{m} (-1)^j \binom{m}{j} \binom{i-jq+m-1}{i-jq}, 
$$
and minimum distance 
$$ 
d=(q-\ell_0)q^{m-\ell_1-1}-1. 
$$
\end{theorem}

Let $\bone=(1,1,\cdots, 1) \in \gf(q)^n$ and 
$$ 
\gf(q)\bone =\{a \bone: a \in \gf(q)\}. 
$$ 
Then $\gf(q)\bone$ is a subspace of $\gf(q)^n$ with dimension 1. A proof of the following 
property can be found in \cite{AssmusKey98}.

\begin{theorem}\label{thm-PGRMcodedual} 
The dual codes $\cPGRM_q(\ell, m)^\perp$ and the original ones $\cPGRM_q(\ell', m)$ are related as follows:  
$$
\cPGRM_q(\ell, m)^\perp = (\gf(q)\bone)^\perp \cap \cPGRM_q(m(q-1)-\ell, m). 
$$
\end{theorem}

When $q=2$, $\cPGRM_q(\ell, m)$ becomes the punctured binary Reed-Muller code. Hence, 
$\cPGRM_q(\ell, m)$ is indeed a generalisation of the original punctured binary Reed-Muller 
code. Other properties of the code $\cPGRM_q(\ell, m)$ can be found in \cite{AssmusKey92} and the book 
chapter \cite{AssmusKey98}. 

The only purpose of introducing the codes $\cPGRM_q(\ell, m)$ in this section is to show the difference between the punctured generalised Reed-Muller codes 
$\cPGRM_q(\ell, m)$ and the new family of generalised Reed-Muller codes to be introduced in 
the next section.

\section{Another generalisation of the punctured binary Reed-Muller codes}\label{sec-newGRMcodes}

\subsection{The newly generalised codes $\cGRM{(q,m, h)}$}

Let $m$ be a positive integer and let $n=q^m-1$, where $q=p^s$, $p$ is a prime 
and $s$ 
is a positive integer. For any integer $a$ with $0 \leq a \leq n-1$, we have the following $q$-adic 
expansion 
$$ 
a=\sum_{j=0}^{m-1} a_jq^j, 
$$
where $0 \leq a_j \leq q-1$. The Hamming weight of $a$, denoted by $\wt(a)$, is the the number 
of nonzero coordinates in the vector $(a_0, a_1, \cdots, a_{m-1})$.

Let $\alpha$ be a generator of $\gf(q^m)^*$. For any $1 \leq h \leq m$, we define a polynomial 
\begin{eqnarray}\label{eqn-generatorplym}
g_{(q,m,h)}(x)=\prod_{\myatop{1 \leq a \leq n-1}{1 \leq \wt(a) \leq h }} (x-\alpha^a).  
\end{eqnarray}
Since $\wt(a)$ is a constant function on each $q$-cyclotomic coset modulo $n$, 
$g_{(q, m, h)}(x)$ is a polynomial over $\gf(q)$. By definition, $g_{(q, m, h)}(x)$ is a 
divisor of $x^n-1$. 

Let $\cGRM{(q,m, h)}$ denote the cyclic code over 
$\gf(q)$ with length $n$ and generator polynomial $g_{(m,q,h)}(x)$.  
By definition, $g_{(q,m,m)}(x)=(x^n-1)/(x-1)$. Therefore, the code $\cGRM{(q,m,m)}$ is trivial. 
Below we consider the code $\cGRM{(q,m, h)}$ for $1 \leq h \leq m-1$ only.

To analyse this code, we set 
\begin{eqnarray}\label{eqn-Index3}
I(q,m, h)=\{1 \leq a \leq n-1: 1 \leq \wt(a) \leq h\}. 
\end{eqnarray} 
The dimension of the code $\cGRM{(q,m, h)}$ is equal to $n-|I(q,m,h)|$.

\begin{theorem}\label{thm-gRMcode} 
Let $m \geq 2$ and $1 \leq h \leq m-1$. Then  
$\cGRM{(q,m,h)}$ has parameters $[q^m-1, k, d]$, where 
$$ 
k=q^m-\sum_{i=0}^{h} \binom{m}{i} (q-1)^i. 
$$ 
and 
\begin{eqnarray}\label{eqn-bbound}
\frac{q^{h +1}-1}{q-1} \leq d \leq 2 q^h-1. 
\end{eqnarray}
\end{theorem}

\begin{proof}
As shown earlier, $I(q,m,h)$ is the union of some $q$-cyclotomic cosets modulo $n$. 
The total number of elements in $\cN$ with Hamming weight $i$ is equal to $\binom{m}{i} (q-1)^i$. 
It then follows that 
$$ 
|I(q,m,h)|=\sum_{i=1}^{h} \binom{m}{i} (q-1)^i. 
$$
Hence, the dimension $k$ of the code is given by 
$$ 
k=q^m-1-\sum_{i=1}^{h} \binom{m}{i} (q-1)^i.
$$ 

Note that every integer $a$ with $1 \leq a \leq (q^{h+1}-1)/(q-1)-1$ has Hamming weight 
$\wt(a) \leq h$. By definition, 
$$ 
\left\{1,2,3, \cdots, (q^{h+1}-1)/(q-1)-1\right\} \subset I(q,m,h). 
$$
It then follows from the BCH bound that $d \geq (q^{h +1}-1)/(q-1)$. 

We now prove the upper bound on the minimum distance $d$ given in (\ref{eqn-bbound}). 
Define  
$$
\ell = (m-h)(q-1)-1 = (q-1)(m-h-1)+q-2=(q-1)\ell_1+\ell_0, 
$$ 
where $\ell_1=m-h-1$ and $\ell_0=q-2$. If $1 \leq a \leq n-1$ and $\wt(a) \leq h$, then 
$\varpi(a) \leq h(q-1) = (q-1)m -\ell -1$. It then follows that $g_{(q, m, h)}(x)$ divides 
$g_{(q, m, \ell)}^R(x)$. Consequently, $\cPGRM{(q,m, \ell)}$ is a subcode of $\cGRM{(q,m,h)}$.  
But by Theorem \ref{thm-PGRMcode}, the minimum distance of $\cPGRM{(q,m, \ell)}$ is equal to 
$$ 
(q-\ell_0)q^{m-\ell_1-1}-1=2q^h-1. 
$$ 
The desired upper bound on $d$ then follows. 
\end{proof}

When $q=2$, the code $\cGRM{(q,m, h)}$ clearly becomes the classical punctured binary 
Reed-Muller code.  Hence, $\cGRM{(q,m, h)}$ is indeed a generalisation of the original 
punctured binary Reed-Muller code. In addition, when $q=2$, the lower bound and the upper 
bound in (\ref{eqn-bbound}) become identical.

\begin{problem}\label{prob-one} 
Is it true that the minimum distance of the code $\cGRM{(q,m,h)}$ is equal to $(q^{h +1}-1)/(q-1)$? 
\end{problem}

\begin{example} 
The following is a list of examples of the code $\cGRM{(q,m,h)}$. 
\begin{itemize}
\item When $(q,m, h)=(3,3,1)$,  $\cGRM{(q,m,h)}$ has parameters $[26,20,4]$. 
\item When $(q,m, h)=(3,4,1)$,  $\cGRM{(q,m,h)}$ has parameters $[80,72,4]$. 
\item When $(q,m, h)=(3,4,2)$,  $\cGRM{(q,m,h)}$ has parameters $[80,48,13]$. 
\item When $(q,m, h)=(3,4,3)$,  $\cGRM{(q,m,h)}$ has parameters $[80,16,40]$.
\item When $(q,m, h)=(4,3,1)$,  $\cGRM{(q,m,h)}$ has parameters $[63,54,5]$. 
\end{itemize}
\end{example}

\subsection{The dual codes $\cGRM{(q,m, h)}^\perp$}

When $q=2$, the parameters of the dual code $\cGRM{(q,m, h)}^\perp$ are given by Theorems \ref{thm-PGRMcode} 
and \ref{thm-PGRMcodedual}. Therefore, we need to study the dual code $\cGRM{(q,m, h)}^\perp$ 
for the case $q>2$ only. 

We will need the following theorem (\cite{HT}, see also \cite[p. 153]{HP03}). 

\begin{theorem}[Hartmann-Tzeng Bound]\label{thm-Hartmann-Tzeng}  
Let $\C$ be a cyclic code of length $n$ over $\gf(q)$ with defining set $T$. Let $A$ be a 
set of $\delta-1$ consecutive elements of $T$ and $B=\{jb \bmod{n}: 0 \leq j \leq s\}$, 
where $\gcd(b, n) < \delta$. If $A+B \subseteq T$, then the minimum distance $d$ of $\C$ 
satisfies $d \geq \delta +s$. 
\end{theorem}

The following theorem gives information on the parameters of the dual code $\cGRM{(q,m,h)}^\perp$.

\begin{theorem}\label{thm-gRMcodedual} 
Let $m \geq 2$ and $1 \leq h \leq m-1$. The dual code 
$\cGRM{(q,m,h)}^\perp$ has parameters $[q^m-1, k^\perp, d^\perp]$, where 
$$ 
k^\perp=\sum_{i=1}^{h} \binom{m}{i} (q-1)^i. 
$$ 
The minimum distance $d^\perp$ of $\cGRM{(q,m,h)}^\perp$ is lower bounded by 
$$ 
d^\perp \geq q^{m-h} + q-2. 
$$
\end{theorem}

\begin{proof}
The desired conclusion on the dimension of $\cGRM{(q,m,h)}^\perp$ follows from the dimension 
of $\cGRM{(q,m,h)}$. What remains to 
be proved is the lower bound on the minimum distance $d^\perp$. Let $\cGRM{(q,m,h)}^c$ 
denote the complement of $\cGRM{(q,m,h)}$, which is generated by the check polynomial 
of $\cGRM{(q,m,h)}$. It is well known that $\cGRM{(q,m,h)}^c$ 
and $\cGRM{(q,m,h)}^\perp$ have the same length, dimension and minimum distance. 

By definition, the defining set of $\cGRM{(q,m,h)}^c$ is 
$$ 
I{(q,m,h)}^c:=\{0\} \cup \{1 \leq b \leq n-1: \wt(b) \geq h +1\}. 
$$ 

Let $b=q^{m-h}+q^{m-h+1}+ \cdots + q^{m-1}$. Define 
$$ 
A=\{a+b: 1 \leq a \leq q^{m-h}-1\} 
$$ 
and 
$$ 
B=\{jb: 0 \leq j \leq q-2\}. 
$$
It is straightforward to verify that $A+B \subset I(q, m, h)^c$. Note that $n 
\in A+B$. In this case, we identify $n$ with $0$.

Clearly, $A$ is a set of $q^{m-h}-1$ consecutive elements in the defining set 
$I(q, m, h)^c$. Note that 
$$ 
\gcd(b,n)=\gcd\left( \frac{q^h-1}{q-1}, q^m-1  \right) 
\leq \gcd(q^h-1, q^m-1) =q^{\gcd(h, m)}-1. 
$$
By assumption, $1 \leq h \leq m-1$. We then have $\gcd(h, m) \leq m-h$. 
Consequently,  
$$ 
\gcd(b,n) < q^{m-h}. 
$$
The desired conclusion on $d^\perp$ then follows from Theorem \ref{thm-Hartmann-Tzeng}. 
\end{proof}

When $q=2$, the lower bound on the minimum distance $d^\perp$ of $\cGRM{(q,m,h)}^\perp$ given 
in Theorem \ref{thm-gRMcodedual} is achieved. It is open if this lower bound can be improved for 
$q>2$.

\begin{problem}\label{prob-two} 
Determine the minimum distance $d^\perp$ of the code $\cGRM{(q,m,h)}^\perp$. 
\end{problem}

To further study the dual code $\cGRM{(q,m, h)}^\perp$, we need to establish relations 
between $\wt(a)$ and $\wt(n-a)$ for $a \in \cN$. Let $a \in \cN$ and let 
$$ 
a=\sum_{j=0}^{m-1} a_j q^j 
$$  
be the $q$-adic expansion of $a$. We define 
$$ 
\gamma(a)=|\{0 \leq j \leq m-1: 1 \leq a_j < q-1 \}|=\wt(a)-|\{0 \leq j \leq m-1: a_j = q-1 \}|. 
$$
Then we have the following lemma whose proof is straightforward and omitted. 

\begin{lemma}\label{lem-april91} 
For $a \in \cN$, we have 
$$ 
\wt(n-a)=m-\wt(a)+\gamma(a)=m-|\{0 \leq j \leq m-1: a_j = q-1 \}|.
$$
\end{lemma}

For $0 \leq i \leq m$, define 
$$ 
N(i)=\{a \in \cN: \wt(a)=i\}  
$$ 
and 
$$ 
-N(i)=\{n-a: a \in N(i)\}.   
$$ 
Clearly, $|N(i)|=\binom{m}{i} (q-1)^i$. 

\vspace{.2cm}
The following lemma will be useful later. 

\begin{lemma}\label{lem-april92} 
In the set $-N(i)$, there are exactly $\binom{m}{i} \binom{i}{j} (q-2)^j$ elements 
with Hamming weight $m-i+j$ for each $j$ with $0 \leq j \leq i$.  
\end{lemma}  

\begin{proof}
Let $a \in N(i)$. By definition, $\wt(a)=i$. It follows from Lemma \ref{lem-april91} that 
$$ 
\wt(n-a)=m-i + \gamma(a). 
$$ 
It is easily seen that 
$$ 
|\{1 \leq a \leq n-1: \wt(a)=i \mbox{ and } \gamma(a)=j\}|=\binom{m}{i} \binom{i}{j} (q-2)^j. 
$$
This completes the proof. 
\end{proof}

\begin{theorem}\label{thm-may12}
$\cGRM{(q,m, h)}^\perp$ is a proper subcode of $\cGRM{(q,m, m-1-h)}$. 
When $q=2$, $\cGRM{(q,m, h)}^\perp$ is the even-weight subcode of $\cGRM{(q,m, m-1-h)}$. 
\end{theorem} 

\begin{proof}
By definition, the defining set of $\cGRM{(q,m, h)}^\perp$ is $-I{(q,m,h)}^c$. We now prove 
that  
$$ 
I{(q,m, m-1-h)} \subset -I{(q,m,h)}^c. 
$$ 
This is equivalent to proving that for every $a \in I{(q,m, m-1-h)}$, $n-a \in I{(q,m,h)}^c$. 
This is clearly true by Lemma \ref{lem-april91}. Consequently, $\cGRM{(q,m, h)}^\perp$ is a proper subcode of $\cGRM{(q,m, m-1-h)}$. 

When $q=2$, we have always the equality that $\wt(a)=m-\wt(n-a)$ for all $a$. Hence, in this case, 
we have 
$$ 
\{0\} \cup I{(q,m, m-1-h)} = -I{(q,m,h)}^c.
$$
As a result, $\cGRM{(q,m, h)}^\perp$ is the even-weight subcode of $\cGRM{(q,m, m-1-h)}$ 
when $q=2$.   
\end{proof}

Experimental data shows that one of $I{(q,m, m-h)}$ and $-I{(q,m,h)}^c$ is not a subset 
of the other. Consequently, one of $\cGRM{(q,m, h)}^\perp$ and $\cGRM{(q,m, m-h)}$ is not 
a subcode of the other.

\begin{example} 
The following is a list of examples of the code $\cGRM{(q,m,h)}^\perp$. 
\begin{itemize}
\item When $(q,m, h)=(2,4,2)$, the code $\cGRM{(q,m,h)}^\perp$ has parameters $[15,10,4]$.  
In this case, the lower bound on the minimum distance is achieved. 
\item When $(q,m, h)=(3,3,1)$, the code $\cGRM{(q,m,h)}^\perp$ has parameters $[26,6,15]$.  
In this case, the lower bound on the minimum distance is $10$. 
\item When $(q,m, h)=(3,3,2)$, the code $\cGRM{(q,m,h)}^\perp$ has parameters $[26,18,6]$.  
In this case, the lower bound on the minimum distance is $4$. 
\end{itemize}
\end{example}

\subsection{The BCH cover of the cyclic code $\cGRM{(q,m,h)}$}

Recall that $n=q^m-1$. 
For any $i$ with $0 \leq i \leq n-1$, let $\m_i(x)$ denote the minimal polynomial of $\alpha^i$ 
over $\gf(q)$. For any $2 \leq \delta \leq n$, define 
\begin{eqnarray}\label{eqn-BCHgeneratorPolyn}
\bar{g}_{(q,n,\delta,b)}(x)=\lcm(\m_{b}(x), \m_{b+1}(x), \cdots, \m_{b+\delta-2}(x)), 
\end{eqnarray} 
where $b$ is an integer, $\lcm$ denotes the least common multiple of these minimal polynomials, and the addition 
in the subscript $b+i$ of $\m_{b+i}(x)$ always means the integer addition modulo $n$. 
Let $\BCH{(q, n, \delta,b)}$ denote the cyclic code of length $n$ with generator 
polynomial $\bar{g}_{(q, n,\delta, b)}(x)$. 

When $b=1$, the set $\BCH{(q, n, \delta, b)}$ is called a \emph{narrow-sense primitive BCH code} with \emph{designed distance} $\delta$. 

The BCH cover of a cyclic code is the BCH code with the smallest dimension containing the cyclic code 
as a subcode. 

\begin{theorem}
$\cGRM{(q,m,h)}$ is a subcode of $\BCH{(q, n, (q^{h +1}-1)/(q-1), 1)}$. 
\end{theorem}

\begin{proof}
In the proof of Theorem \ref{thm-gRMcode}, we have shown that 
$$ 
\{1,2,3, \cdots, (q^{h+1}-1)/(q-1)-1\} \subset I(q,m,h). 
$$ 
Hence, the generator polynomial of $\BCH{(q, n, (q^{h +1}-1)/(q-1), 1)}$ is a divisor of 
that of $\cGRM{(q,m,h)}$. Hence, $\cGRM{(q,m,h)}$ is a subcode of  $\BCH{(q, n, (q^{h +1}-1)/(q-1), 1)}$. 
\end{proof}

When $h =1$ or $h=m-1$, the two codes are identical. In other cases, the dimension of 
the code $\BCH{(q, n, (q^{h +1}-1)/(q-1), 1)}$ is larger than that of $\cGRM{(q,m,h)}$.  

The BCH cover of $\cGRM{(q,m,h)}$ is $\BCH{(q, n, (q^{h +1}-1)/(q-1), 1)}$ if the 
minimum distance of $\cGRM{(q,m,h)}$ is indeed equal to $(q^{h +1}-1)/(q-1)$.

\subsection{Comparisons of the two codes $\cPGRM_q(\ell, m)$ and $\cGRM(q,m,h)$} 

In this subsection, we compare the two codes $\cPGRM_q(\ell, m)$ and $\cGRM(q,m,h)$ 
and make some comments. 

First of all, the two codes $\cPGRM_q(\ell, m)$  and $\cGRM(q,m,h)$ are clearly different, 
as their dimensions and minimum distances are different. Secondly, Theorem \ref{thm-may12} 
tells us that $\cGRM(q,m,h)^\perp$ is indeed a subcode of $\cGRM(q,m,m-1-h)$. But the 
code $\cGRM(q,m,h)$ does not have the property of Theorem \ref{thm-PGRMcodedual}.

\section{The extended codes $\widehat{\cGRM}(q,m,h)$ and their applications in combinatorics} 

Let $\widehat{\cGRM}(q,m,h)$ denote the extended code of $\cGRM(q,m,h)$ defined in Section 
\ref{sec-newGRMcodes}. The following theorem follows directly from Theorem \ref{thm-gRMcode}. 

\begin{theorem}\label{thm-gRMcodeExt} 
Let $m \geq 2$ and $1 \leq h \leq m-1$. Then  
$\widehat{\cGRM}{(q,m,h)}$ has parameters $[q^m, k, d]$, where 
$$ 
k=q^m-\sum_{i=0}^{h} \binom{m}{i} (q-1)^i. 
$$ 
and 
\begin{eqnarray}\label{eqn-bbound2}
\frac{q^{h +1}-1}{q-1} \leq d \leq 2 q^h. 
\end{eqnarray}
\end{theorem} 

The code $\widehat{\cGRM}(q,m,h)$ is the newly generalised Reed-Muller code. Our objective in this 
section is to prove that $\widehat{\cGRM}(q,m,h)$ is affine-invariant and demonstrates its 
application in combinatorics.    

\subsection{The code $\widehat{\cGRM}(q,m,h)$ is affine-invariant} 

\subsubsection{Automorphism groups of linear codes}

The set of coordinate permutations that map a code $\C$ to itself forms a group, which is referred to as 
the \emph{permutation automorphism group\index{permutation automorphism group of codes}} of $\C$
and denoted by $\PAut(\C)$. If $\C$ is a code of length $n$, then $\PAut(\C)$ is a subgroup of the 
\emph{symmetric group\index{symmetric group}} $\Sym_n$.

A \emph{monomial matrix\index{monomial matrix}} over $\gf(q)$ is a square matrix having exactly one 
nonzero element of $\gf(q)$  in each row and column. A monomial matrix $M$ can be written either in 
the form $DP$ or the form $PD_1$, where $D$ and $D_1$ are diagonal matrices and $P$ is a permutation 
matrix. 

The set of monomial matrices that map $\C$ to itself forms the group $\MAut(\C)$,  which is called the 
\emph{monomial automorphism group\index{monomial automorphism group}} of $\C$. Clearly, we have 
$$
\PAut(\C) \subseteq \MAut(\C).
$$

The \textit{automorphism group}\index{automorphism group} of $\C$, denoted by $\GAut(\C)$, is the set 
of maps of the form $M\gamma$, 
where $M$ is a monomial matrix and $\gamma$ is a field automorphism, that map $\C$ to itself. In the binary 
case, $\PAut(\C)$,  $\MAut(\C)$ and $\GAut(\C)$ are the same. If $q$ is a prime, $\MAut(\C)$ and 
$\GAut(\C)$ are identical. In general, we have 
$$ 
\PAut(\C) \subseteq \MAut(\C) \subseteq \GAut(\C). 
$$

By definition, every element in $\GAut(\C)$ is of the form $DP\gamma$, where $D$ is a diagonal matrix, 
$P$ is a permutation matrix, and $\gamma$ is an automorphism of $\gf(q)$.   
The automorphism group $\GAut(\C)$ is said to be $t$-transitive if for every pair of $t$-element ordered 
sets of coordinates, there is an element $DP\gamma$ of the automorphism group $\GAut(\C)$ such that its 
permutation part $P$ sends the first set to the second set. 

It is in general very difficult to determine the full automorphism group $\GAut(\C)$ of a linear code $\C$. 
However, for many applications it is sufficient to find a proper subgroup of the $\GAut(\C)$. We will do 
this for the code $\widehat{\cGRM}(q,m,h)$ subsequently. 

\subsubsection{Affine-invariant linear codes} 

In this section, we first give a special representation of primitive cyclic codes and their extended codes, 
and then define and characterise affine-invariant codes. We will skip proof details, but 
refer the reader to \cite[Section 4.7]{HP03} for a detailed proof of the major results presented in this section. 

A cyclic code of length $n=q^m-1$ over $\gf(q)$ for some positive integer $m$ is called a \emph{primitive 
cyclic code}. Let $\cR_n$ denote the quotient ring $\gf(q)[x]/(x^n-1)$. Any cyclic code $\C$ of length 
$n=q^m-1$ over $\gf(q)$ is an ideal of $\cR_n$, and is generated by a monic polynomial $g(x)$ of the 
least degree over $\gf(q)$. This polynomial is called the generator polynomial of the cyclic code $\C$, and 
can be expressed as 
$$ 
g(x)=\prod_{t \in T} (x-\alpha^t), 
$$ 
where $\alpha$ is a generator of $\gf(q^m)^*$, $T$ is a subset of $\cN=\{0,1, \cdots, n-1\}$ and a union 
of some $q$-cyclotomic cosets modulo $n$. 
The set $T$ is called a \emph{defining set} of $\C$ with respect to $\alpha$. 
When $\C$ is viewed as a subset of $\cR_n$, 
every codeword of $\C$ is a polynomial 
$c(x)=\sum_{i=0}^{n-1} c_i x^i$, where all 
$c_i \in \gf(q)$. 
A primitive cyclic code $\C$ is called \emph{even-like} if $1$ is a zero of its generator polynomial, 
and \emph{odd-like} otherwise. 

Let $\cJ$ and $\cJ^*$ denote $\gf(q^m)$ and $\gf(q^m)^*$, respectively. Let $\alpha$ be a primitive element 
of $\gf(q^m)$. The set $\cJ$ will be the index set of the extended cyclic codes of length $q^m$, and the set 
$\cJ^*$ will be the index set of the cyclic codes of length $n$. Let $X$ be an indeterminate. Define 
\begin{eqnarray}
\gf(q)[\cJ]=\left\{a=\sum_{g \in \cJ} a_g X^g: a_g \in \gf(q) \mbox{ for all } g \in \cJ \right\}. 
\end{eqnarray}    
The set $\gf(q)[\cJ]$ is an algebra under the following operations 
\begin{eqnarray*}
u \sum_{g \in \cJ} a_g X^g + v \sum_{g \in \cJ} b_g X^g = \sum_{g \in \cJ} (ua_g +v b_g) X^g  
\end{eqnarray*} 
for all $u, \, v \in \gf(q)$, and 
\begin{eqnarray}
\left(\sum_{g \in \cJ} a_g X^g \right) \left(\sum_{g \in \cJ} b_g X^g \right) 
= \sum_{g \in \cJ} \left(\sum_{h \in \cJ} a_h b_{g-h} \right) X^g. 
\end{eqnarray} 
The zero and unit of $\gf(q)[\cJ]$ are $\sum_{g \in \cJ} 0 X^g$ and $X^0$, respectively. 

Similarly, let 
\begin{eqnarray}
\gf(q)[\cJ^*]=\left\{a=\sum_{g \in \cJ^*} a_g X^g: a_g \in \gf(q) \mbox{ for all } g \in \cJ^* \right\}. 
\end{eqnarray}   
The set $\gf(q)[\cJ^*]$ is not a subalgebra, but a subspace of $\gf(q)[\cJ]$. Obviously, the elements of 
$\gf(q)[\cJ^*]$ are of the form 
$$ 
\sum_{i=0}^{n-1} a_{\alpha^i} X^{\alpha^i}, 
$$  
and those of 
$\gf(q)[\cJ]$ are of the form 
$$ 
a_0X^0 + \sum_{i=0}^{n-1} a_{\alpha^i} X^{\alpha^i}.  
$$  
Subsets of the subspace $\gf(q)[\cJ^*]$ will be used to characterise primitive cyclic codes over $\gf(q)$ and 
those of the algebra $\gf(q)[\cJ]$ will be employed to characterise extended primitive cyclic codes over $\gf(q)$. 

We define a one-to-one correspondence between $\cR_n$ and $\gf(q)[\cJ^*]$ by 
\begin{eqnarray}\label{eqn-Upsilon}
\Upsilon: c(x)=\sum_{i=0}^{n-1} c_i x^i \to C(X)=\sum_{i=0}^{n-1} C_{\alpha^i} X^{\alpha^i},
\end{eqnarray} 
where $C_{\alpha^i}=c_i$ for all $i$. 

The following theorem is obviously true. 

\begin{theorem}\label{thm-newCharacterisationCyclicCodes}
$\C \subseteq \cR_n$ has the circulant cyclic shift property if and only if $\Upsilon(\C) \subseteq \gf(q)[\cJ^*]$ has the property 
that 
$$ 
\sum_{i=0}^{n-1} C_{\alpha^i} X^{\alpha^i} = \sum_{g \in \cJ^*} C_g X^g \in \Upsilon(\C)  
$$ 
if and only if 
$$ 
\sum_{i=0}^{n-1} C_{\alpha^i} X^{\alpha \alpha^i} = \sum_{g \in \cJ^*} C_g X^{\alpha g} \in \Upsilon(\C)  
$$ 
\end{theorem} 

With Theorem \ref{thm-newCharacterisationCyclicCodes}, every primitive cyclic code over $\gf(q)$ can be  
viewed as a special 
subset of $\gf(q)[\cJ^*]$ having the property documented in this theorem. This special representation of 
primitive cyclic codes over $\gf(q)$ will be very useful for determining a subgroup of the automorphism group 
of certain primitive cyclic codes.  

It is now time to extend primitive cyclic codes, which are subsets of $\gf(q)[\cJ^*]$. We use the element $0 
\in \cJ$ to index the extended coordinate. 
The extended codeword 
$\overline{C}(X)$ of a codeword $C(X)=\sum_{g \in \cJ^*} C_g X^g$ in $\gf(q)[\cJ^*]$ is defined by 
\begin{eqnarray}
\overline{C}(X)=\sum_{g \in \cJ} C_g X^g
\end{eqnarray} 
with $\sum_{g \in \cJ} C_g=0.$

Notice that $X^{\alpha 0}=X^0=1$. The following then follows from Theorem \ref{thm-newCharacterisationCyclicCodes}. 

\begin{theorem}\label{thm-newCharacterisationExtendCyclicCodes}
The extended code $\overline{\C}$ of a cyclic code $\C \subseteq \gf(q)[\cJ^*]$ is a subspace of 
$\gf(q)[\cJ]$ such that 
\begin{eqnarray*}
 \overline{C}(X)=\sum_{g \in \cJ} C_g X^g \in \overline{\C} \mbox{ if and only if } 
 \sum_{g \in \cJ} C_g X^{\alpha g} \in \overline{\C} \mbox{ and } 
\sum_{g \in \cJ} C_g=0. 
\end{eqnarray*} 
\end{theorem}

If a cyclic code $\C$ is viewed as an ideal of $\cR_n=\gf(q)[x]/(x^n-1)$, it can be defined by its set of zeros 
or its defining set. When $\C$ and $\overline{\C}$ are put in the 
settings $\gf(q)[\cJ^*]$ and $\gf(q)[\cJ]$, respectively, they can be defined with some counterpart of the defining set. 
This can be done with the assistance of the following function $\phi_s$ from $\gf(q)[\cJ]$ to $\cJ$: 
\begin{eqnarray}
\phi_s\left(\sum_{g \in \cJ} C_g X^g\right)= \sum_{g \in \cJ} C_g g^s, 
\end{eqnarray}    
where $s \in \overline{\cN}:=\{i: 0 \leq i \leq n\}$ and by convention $0^0=1$ in $\cJ$. 

The following follows from Theorem \ref{thm-newCharacterisationExtendCyclicCodes} and the definition of 
$\phi_s$ directly. 

\begin{lemma}\label{lem-dec101}  
$\overline{C}(X)$ is the extended codeword of $C(X) \in \gf(q)[\cJ^*]$ if and only if $\phi_0(\overline{C}(X))=0$. 
In particular, if $\overline{\C}$ is the extended code of a primitive cyclic code $\C \subseteq \gf(q)[\cJ^*]$, 
then $\phi_0(\overline{C}(X))=0$ for all $\overline{C}(X) \in \overline{\C}$. 
\end{lemma}

\begin{lemma}\label{lem-dec102} 
Let $\C$ be a primitive cyclic code of length $n$ over $\gf(q)$. Let $T$ be the defining set of $\C$ with 
respect to $\alpha$, when it is viewed as an ideal of $\cR_n$. Let $s \in T$ and $1 \leq s \leq n-1$.  
We have 
then $\phi_s(\overline{C}(X))=0$ for all $\overline{C}(X) \in \overline{\C}$. 
\end{lemma}

\begin{lemma}\label{lem-dec103}  
Let $\C$ be a primitive cyclic code of length $n$ over $\gf(q)$. Let $T$ be the defining set of $\C$ with 
respect to $\alpha$, when it is viewed as an ideal of $\cR_n$. Then $0 \in T$ if and only if   
$\phi_n(\overline{C}(X))=0$ for all $\overline{C}(X) \in \overline{\C}$. 
\end{lemma}

Combining Lemmas \ref{lem-dec101}, \ref{lem-dec102}, \ref{lem-dec103} and the discussions above, 
we can define an extended cyclic code in terms of a defining set as follows. 

A code $\overline{\C}$ of length $q^m$ is an \emph{extended primitive cyclic code}\index{extended primitive cyclic code} 
with definition set $\overline{T}$ provided $\overline{T} \setminus \{n\} \subseteq \overline{\cN}$ is a union of 
$q$-cyclotomic cosets modulo $n=q^m-1$ with $0 \in \overline{T}$ and 
\begin{eqnarray}\label{eqn-defdefCodes}
\overline{\C}= \left\{\overline{C}(X) \in \gf(q)[\cJ]: \phi_s(\overline{C}(X))=0 \mbox{ for all } s \in \overline{T} \right\}. 
\end{eqnarray} 

The following remarks are helpful for fully understanding the characterisation of extended primitive 
cyclic codes: 
\begin{itemize} 
\item The condition that $\overline{T} \setminus \{n\} \subseteq \overline{\cN}$ is a union of 
$q$-cyclotomic cosets modulo $n=q^m-1$ is to ensure that the code $\C$ obtained by puncturing the 
first coordinate of $\overline{\C}$ and ordering the elements of $\cJ$ with $(0, \alpha^n, \alpha^1, \cdots, \alpha^{n-1})$ is a primitive cyclic code. 
\item The additional requirement $0 \in \overline{T}$ and (\ref{eqn-defdefCodes}) are to make sure 
that $\overline{\C}$ is the extended code of $\C$. 
\item If $n \in \overline{T}$, then $\C$ is an even-like code. In this case, the extension is trivial, 
i.e., the extended coordinate in every codeword of $\overline{\C}$ is always equal to $0$. 
If $n \not\in \overline{T}$, then $0 \not\in T$. Thus, the extension is nontrivial. 
\item If $\overline{\C}$ is the extended code of a primitive cyclic code $\C$, then 
\begin{eqnarray*}
\overline{T}=\left\{ 
\begin{array}{ll}
\{0\} \cup T & \mbox{ if } 0 \not\in T, \\
\{0, n\} \cup T & \mbox{ if } 0 \in T.  
\end{array}
\right.
\end{eqnarray*}
where $T$ and $\overline{T}$ are the defining sets of $\C$ and $\overline{\C}$, respectively.  
\item The following diagram illustrates the relations among the two codes and their definition sets: 
\begin{eqnarray*}
\begin{array}{lcr}
\C \subseteq \cR_n & \Longleftrightarrow \C \subseteq \gf(q)[\cJ^*] 
                      \Longrightarrow & \gf(q)[\cJ] \supseteq \overline{\C} \\
T \subseteq \cN     &      &    \overline{T} \subseteq \overline{\cN}                       
\end{array} 
\end{eqnarray*}
\end{itemize}   

Let $\sigma$ be a permutation on $\cJ$. This permutation acts on a code $\overline{\C} \subseteq \gf(q)[\cJ]$ as follows: 
\begin{eqnarray}
\sigma\left( \sum_{g \in \cJ} C_g X^g \right) = \sum_{g \in \cJ} C_g X^{\sigma(g)}.   
\end{eqnarray}

The \emph{affine permutation group}\index{affine permutation group}, denoted by $\AGL(1, q^m)$, is 
defined by 
\begin{eqnarray}\label{eqn-AGL(1qm)}
\AGL(1, q^m)=\{\sigma_{(a,b)}(y)=ay+b: a \in \cJ^*, \, b \in \cJ\}.  
\end{eqnarray} 
We have the following conclusions about $\AGL(1, q^m)$ whose proofs are straightforward: 
\begin{itemize} 
\item $\AGL(1, q^m)$ is a permutation group on $\cJ$ under the function composition. 
\item The group action of $\AGL(1, q^m)$ on $\gf(q^m)$ is doubly transitive, i.e., $2$-transitive. 
\item $\AGL(1, q^m)$ has order $(n+1)n=q^m(q^m-1)$. 
\item Obviously, the maps $\sigma_{(a,0)}$ are merely the cyclic shifts on the coordinates 
$(\alpha^n, \alpha^1, \cdots, \alpha^{n-1})$ each fixing the coordinate $0$. 
\end{itemize} 

An \emph{affine-invariant code}\index{affine-invariant} is an extended primitive cyclic code 
$\overline{\C}$ such that $\AGL(1, q^m) \subseteq \PAut(\overline{\C})$. For certain applications, 
it is important to know if a given extended primitive cyclic code $\overline{\C}$ is affine-invariant or not. 
This question can be answered by examining the defining set of the code. In order to do this, we introduce a partial 
ordering $\preceq$ on $\overline{\cN}$. Suppose that $q=p^t$ for some positive integer $t$. Then by 
definition $\overline{\cN}=\{0,1,2, \cdots, n\}$, where $n=q^m-1=p^{mt}-1$. The $p$-adic expansion of 
each $s \in \overline{\cN}$ is given by 
\begin{eqnarray*}
s=\sum_{i=0}^{mt-1} s_i p^i, \, \mbox{ where } 0 \leq s_i <p \mbox{ for all } 0 \leq i \leq mt-1.   
\end{eqnarray*}    
Let the $p$-adic expansion of $r \in \overline{\cN}$ be 
$$ 
r=\sum_{i=0}^{mt-1} r_i p^i. 
$$  
We say that $r \preceq s$ if $r_i \leq s_i$ for all $0 \leq i \leq mt-1$. By definition, we have 
$r\leq s$ if $r \preceq s$.

The following is a characterisation of affine-invariant codes due to Kasami, Lin and Peterson \cite{KLP68}. 

\begin{theorem}[Kasami-Lin-Peterson]\label{thm-KLPtheorem}
Let $\overline{\C}$ be an extended cyclic code of length $q^m$ over $\gf(q)$ with defining set $\overline{T}$. 
The code $\overline{\C}$ is affine-invariant if and only if whenever $s \in \overline{T}$ then $r \in \overline{T}$ for all 
$r \in \overline{\cN}$ with $r \preceq s$.    
\end{theorem}

Theorem \ref{thm-KLPtheorem} will be employed in the next section. It is a very useful tool to prove that 
an extended primitive cyclic code is affine-invariant.

\subsubsection{$\widehat{\cGRM}(q,m,h)$ is affine-invariant}

\begin{theorem}\label{thm-OmegaCodeAffineInvariant}
Let $q$ be a prime and $1 \leq h \leq m-1$. Then $\widehat{\cGRM}{(q,m, h)}$ is affine-invariant. 
\end{theorem}

\begin{proof}
By (\ref{eqn-generatorplym}), the defining set $\overline{T}$ of $\widehat{\cGRM}{(q,m, h)}$ is 
given by 
$$ 
\overline{T}=\{0\} \cup \{1 \leq a \leq n-1: 1 \leq \wt(a) \leq h\}. 
$$
Recall that $\overline{\cN}=\{0,1, \cdots, n\}$. Let $s \in \overline{T}$ and $r \in \overline{\cN}$. Assume 
that $r \preceq s$. We need prove that $r \in \overline{T}$ by Theorem \ref{thm-KLPtheorem}. 

If $s=0$, then $r=0$. In this case, we have indeed $r  \in \overline{T}$. We now assume that 
$s>0$ and $r>0$. Since $r \preceq s$, we have 
$$ 
1 \leq \wt(r) \leq \wt(s) < h.  
$$ 
This means that $r \in \overline{T}$. 
The desired conclusion then follows.   
\end{proof} 

\begin{corollary}\label{cor-Omegacodes2Transitive} 
The automorphism group $\Aut(\widehat{\cGRM}{(q,m, h)})$ is doubly transitive. 
\end{corollary} 

\begin{proof}
Recall that the group action of $\AGL(1, q^m)$ on $\gf(q^m)$ is doubly transitive. 
By Theorem \ref{thm-OmegaCodeAffineInvariant}, $\AGL(1, q^m) \subseteq \Aut(\widehat{\cGRM}{(q,m, h)})$. 
Recall that the coordinates of $\Aut(\widehat{\cGRM}{(q,m, h)})$ are indexed by the elements in 
$\gf(q^m)$. The desired conclusion then follows. 
\end{proof}

\subsection{$\widehat{\cGRM}(q,m,h)$ holds $2$-designs}

Let $\cP$ be a set of $v \ge 1$ elements, and let $\cB$ be a set of $k$-subsets of $\cP$, where $k$ is
a positive integer with $1 \leq k \leq v$. Let $t$ be a positive integer with $t \leq k$. The pair
$\bD = (\cP, \cB)$ is called a $t$-$(v, k, \lambda)$ {\em design\index{design}}, or simply {\em $t$-design\index{$t$-design}}, if every $t$-subset of $\cP$ is contained in exactly $\lambda$ elements of
$\cB$. The elements of $\cP$ are called points, and those of $\cB$ are referred to as blocks.
We usually use $b$ to denote the number of blocks in $\cB$.  A $t$-design is called {\em simple\index{simple}} if $\cB$ does not contain repeated blocks. In this paper, we consider only simple 
$t$-designs.  A $t$-design is called {\em symmetric\index{symmetric design}} if $v = b$. It is clear that $t$-designs with $k = t$ or $k = v$ always exist. Such $t$-designs are {\em trivial}. In this paper, we consider only $t$-designs with $v > k > t$.
A $t$-$(v,k,\lambda)$ design is referred to as a {\em Steiner system\index{Steiner system}} if $t \geq 2$ and $\lambda=1$, and is denoted by $S(t,k, v)$. 

Let $\C$ be a $[v, \kappa, d]$ linear code over $\gf(q)$. Let $A_i:=A_i(\C)$, which denotes the
number of codewords with Hamming weight $i$ in $\C$, where $0 \leq i \leq v$. The sequence 
$(A_0, A_1, \cdots, A_{v})$ is
called the \textit{weight distribution} of $\C$, and $\sum_{i=0}^v A_iz^i$ is referred to as
the \textit{weight enumerator} of $\C$. For each $k$ with $A_k \neq 0$,  let $\cB_k$ denote
the set of the supports of all codewords with Hamming weight $k$ in $\C$, where the coordinates of a codeword
are indexed by $(0,1,2, \cdots, v-1)$. Let $\cP=\{0, 1, 2, \cdots, v-1\}$.  The pair $(\cP, \cB_k)$
may be a $t$-$(v, k, \lambda)$ design for some positive integer $\lambda$, which is called a 
\emph{support design} of the code. In such a case, we say that the code $\C$ holds a $t$-$(v, k, \lambda)$ 
design.

A proof of the following theorem can be found in \cite[p. 308]{HP03}. 
 
\begin{theorem}\label{thm-designCodeAutm}
Let $\C$ be a linear code of length $n$ over $\gf(q)$ where $\GAut(\C)$ is $t$-transitive. Then the codewords of any weight $i \geq t$ of $\C$ hold a $t$-design.
\end{theorem}

The task of this section is to prove the following theorem. 

\begin{theorem}\label{thm-2designOmegaCode}
Let $m \geq 2$ and $1 \leq h \leq m-1$. Let $\hat{A}_i$ denote the number of codewords of weight $i$ in 
$\widehat{\cGRM}(q,m,h)$. Then for each $i$ with $\hat{A}_i \neq 0$, the supports of the codewords of 
weight $i$ form a $2$-design.   
\end{theorem} 

\begin{proof}
The desired conclusion follows from Theorem \ref{thm-designCodeAutm} and Corollary 
\ref{cor-Omegacodes2Transitive}. 
\end{proof}

\begin{example} 
Let $(q,m,h)=(3,3,2)$. Then $\widehat{\cGRM}(q,m,h)$ has parameters $[27, 8, 14]$ and weight enumerator 
$$ 
1 + 810z^{14} + 702z^{15} + 1404z^{17} + 780z^{18} + 2106z^{20} + 702z^{21} +  54z^{26} + 2z^{27}. 
$$ 
It holds $2$-$(27, k, \lambda)$ designs for the following pairs $(k, \lambda)$: 
$$ 
(14, 105), \ (15, 105), \ (17, 272), \ (18, 170), \ (20, 570), \ (21, 210).    
$$
\end{example} 

We remark that $\widehat{\cGRM}(q,m,h)$ does not hold $3$-designs in general. Of course, 
it holds $3$-designs when $q=2$, as $\widehat{\cGRM}(q,m,h)$ is identical with the classical 
binary Reed-Muller code when $q=2$. Though $\widehat{\cGRM}(q,m,h)$ holds many $2$-$(q^m, k, \lambda)$ 
designs, determining the parameters $k$ and $\lambda$ may be hard. The reader is cordially invited 
to attack this problem.

\section{A family of reversible cyclic codes from the codes $\cGRM(q,m,h)$}

A code $\C$ is called {\em reversible} if $(c_0, c_1, \ldots, c_{n-1}) \in \C$ implies that  
$(c_{n-1}, c_{n-2}, \ldots, c_{0})$ $\in \C$. 

Let $f(x)=f_{j}x^j+f_{j-1}x^{j-1}+ \cdots + f_1x+f_0$ be a polynomial over $\gf(q)$ with $f_j \ne 0$ 
and $f_0 \ne 0$. The reciprocal $f^*(x)$ of $f(x)$ is defined by 
$$ 
f^*(x)=f_0^{-1}x^{j}f(x^{-1}).  
$$ 
A polynomial $f(x)$ is called \emph{self-reciprocal} if $f$ and its reciprocal are identical.  

The conclusions of the following theorem are known in the literature, and are easy to prove. We will employ 
some of them later.  

\begin{theorem}\label{thm-ReversibleCyclicCodes}
Let $\C$ be a cyclic code over $\gf(q)$ with generator polynomial $g(x)$. Then the following statements are equivalent. 
\begin{itemize}
\item $\C$ is reversible. 
\item $g$ is self-reciprocal. 
\item $\beta^{-1}$ is a root of $g$ for every root $\beta$ of $g(x)$ over the splitting field of $g(x)$.
\end{itemize}
\end{theorem} 

If $\C$ is a reversible cyclic code of length $n$ over $\gf(q)$, then $\C \oplus \C^\perp=\gf(q)^n$. 
Such a linear code is called a \emph{linear code with complement dual (LCD)}, as its dual code is equal 
to its complement. 

LCD cyclic codes over finite fields are interesting in both theory and applications 
\cite{Massey64,Massey92,Sendr}. An important application of LCD codes in cryptography 
was recently documented in \cite{CG}. This is our major motivation of constructing LCD 
codes.

We now employ the codes $\cGRM{(q,m,h)}$ to construct reversible cyclic codes. 
To this end, we need to make some preparations. 

Recall that  
\begin{eqnarray}\label{eqn-Index2}
I{(q,m,t)}=\{1 \leq i \leq n-1: 1 \leq \wt(i) \leq t \}   
\end{eqnarray} 
and 
$$
-I{(q,m,t)}=\{n-a: a \in I_{(q,m,t)}\},  
$$
where $1 \leq t \leq m$.

\begin{lemma}\label{lem-RMc} 
If $1 \leq t \leq \lceil m/2 \rceil-1$, then $I{(q,n,t)} \cap (-I{(q,n,t)}) = \emptyset$. 
\end{lemma} 

\begin{proof}
Note that 
$$n=q^m-1=(q-1)q^{m-1} + (q-1)q^{m-2} + \cdots + (q-1)q+(q-1)q^0. 
$$ 
By Lemma \ref{lem-april91}, $\wt(n-i) \geq m -\wt(i)$ for all $i \in \cN$. 

If $i \in \cN$ and $\wt(i) \leq \lceil m/2 \rceil-1$, then 
$\wt(n-i)  \geq m -\wt(i) > \lceil m/2 \rceil-1$. 
The desired conclusion then follows. 
\end{proof}

Let $g_{(q, m, h)}(x)$ be the polynomial of \eqref{eqn-generatorplym}, which is the generator 
polynomial of the cyclic code $\cGRM{(q,m,h)}$. 
Let $g^*_{(q, m, h)}(x)$ denote the reciprocal of $g_{(q, m, h)}(x)$. Set 
$$ 
g(x)=(x-1)\lcm\left(g_{(q, m, h)}(x), g^*_{(q, m, h)}(x)\right). 
$$
Let $\overline{\cGRM}{(q, m, h)}$ denote the cyclic code of length $n$ over $\gf(q)$ with generator polynomial $g(x)$. It follows from Theorem \ref{thm-ReversibleCyclicCodes} that $\overline{\cGRM}{(q, m, h)}$ is reversible.

Information on the parameters of the reversible cyclic code $\overline{\cGRM}{(q, m, h)}$ is given 
in the theorem below. 

\begin{theorem}\label{thm-bCode}
If $1 \leq h \leq \lceil m/2 \rceil-1$, then the reversible cyclic code $\overline{\cGRM}{(q, m, h)}$ has minimum distance 
$$d \geq 2\frac{q^{h +1}-1}{q-1}$$ 
and dimension 
\begin{eqnarray}
q^m-2 \sum_{i=0}^{h} \binom{m}{i} (q-1)^i. 
\end{eqnarray}

\end{theorem}  

\begin{proof}
When $1 \leq h \leq \lceil m/2 \rceil-1$,   
it follows from Lemma \ref{lem-RMc} that $g_{(q, m, h)}(x)$ and $g^*_{(q, m, h)}(x)$ are relatively prime. Consequently, 
$g(x)=(x-1)g_{(q, m, h)}(x) g^*_{(q, m, h)}(x)$. Therefore,  
$$ 
\deg(g(x))=2\deg(g_{(q, m, h)}(x))+1. 
$$ 
By Theorem \ref{thm-gRMcode}, 
$$ 
\deg(g_{(q, m, h)}(x))=\sum_{i=1}^{h} \binom{m}{i} (q-1)^i. 
$$ 
The desired conclusion on the dimension then follows. 

In this case, it follows from the proof of Theorem \ref{thm-gRMcode} that $g(x)$ 
has the roots $\alpha^i$ for all $i$ in the set 
$$ 
\left\{n-\left(\frac{q^{h +1}-1}{q-1}-1\right), \cdots, n-2, n-1, 0, 1, 2, \cdots, \frac{q^{h +1}-1}{q-1}-1\right\}. 
$$   
The desired conclusion on the minimum distance then follows from the BCH bound. 
\end{proof}

\begin{example} 
The following is a list of examples of the reversible cyclic code $\overline{\cGRM}{(q, m, h)}$. 
\begin{itemize}
\item When $(q,m, h)=(2,4,1)$, the code $\overline{\cGRM}{(q, m, h)}$ has parameters $[15,6,6]$. 
\item When $(q,m, h)=(2,6,2)$, the code $\overline{\cGRM}{(q, m, h)}$ has parameters $[63,20,14]$. 
\item When $(q,m, h)=(3,4,1)$, the code $\overline{\cGRM}{(q, m, h)}$ has parameters $[80,63,8]$. 
\end{itemize}
\end{example}

\begin{problem}\label{prob-three} 
Determine the minimum distance of the code $\overline{\cGRM}{(q, m, h)}$ of 
Theorem \ref{thm-bCode}. 
\end{problem}

\begin{theorem}\label{thm-bCode2}
Let $m \geq 2$ be even. 
Then the reversible cyclic code $\overline{\cGRM}{(q, m, m/2)}$ has minimum distance 
$$d \geq 2\frac{q^{(m+2)/2}-1}{q-1}$$ 
and dimension 
\begin{eqnarray}
q^m-2 \sum_{i=0}^{m/2} \binom{m}{i} (q-1)^i  + \binom{m}{\frac{m}{2}}. 
\end{eqnarray}
\end{theorem} 

\begin{proof}
The conclusion on the minimum distance comes from the BCH bound. We now prove the conclusion on 
the dimension of the code. It follows from Lemmas \ref{lem-RMc}, \ref{lem-april91} and \ref{lem-april92} that 
$a \in I{(q,m,m/2)} \cap (-I{(q,m,m/2)})$ if and only if $\wt(a)=m/2$ and the $q$-adic expression 
of $a$ is of the form 
$$ 
(q-1)(q^{i_1}+q^{i_2} + \cdots + q^{i_{m/2}}), 
$$
where $0 \leq i_1 < i_2 < \cdots < i_{m/2} \leq m-1$. Consequently, 
$$ 
|I{(q,m,m/2)} \cap (-I{(q,m,m/2)})|=\binom{m}{\frac{m}{2}}. 
$$

As before, let $g_{(q, m, m/2)}(x)$ be the generator polynomial of the code $\cGRM{(q,m, m/2)}$. 
Then the generator polynomial of $\overline{\cGRM}{(q, m, m/2)}$ is given by 
\begin{eqnarray*}
g(x) &=& (x-1)\lcm\left(g_{(q, m, m/2)}(x),\, g^*_{(q, m, m/2)}(x)\right) \\
&=& \frac{(x-1)g_{(q, m, m/2)}(x)g^*_{(q, m, m/2)}(x)}{\gcd\left(g_{(q, m, m/2)}(x), \, g^*_{(q, m, m/2)}(x)\right)}.  
\end{eqnarray*}
Therefore, 
\begin{eqnarray*}
\deg(g(x)) &=& 2\deg(g_{(q, m, m/2)}(x))+1-\deg(\gcd(g_{(q, m, m/2)}(x), g^*_{(q, m, m/2)}(x)) \\
&=& 2\deg(g_{(q, m, m/2)}(x))+1- |I{(q,m,m/2)} \cap (-I{(q,m,m/2)})| \\ 
&=& 1+2 \sum_{i=1}^{m/2} \binom{m}{i} (q-1)^i  - \binom{m}{\frac{m}{2}}. 
\end{eqnarray*}
The desired conclusion on the dimension then follows. 
\end{proof} 

We point out that the dimension of the code $\overline{\cGRM}{(q, m, m/2)}$ is equal to zero 
when $q=2$. Hence, the code is nontrivial only when $q>2$. 

\begin{example} 
When $(q,m)=(5,2)$, the code $\overline{\cGRM}{(q, m, m/2)}$ has parameters $[24,9,12]$.  
\end{example}

\begin{problem}\label{prob-four} 
Determine the minimum distance of the code $\overline{\cGRM}{(q, m, m/2)}$ of 
Theorem \ref{thm-bCode2}. 
\end{problem}

\begin{theorem}\label{thm-bCode3}
Let $m \geq 3$ be odd. 
Then the reversible cyclic code $\overline{\cGRM}{(q, m, (m+1)/2)}$ has minimum distance 
$$d \geq 2\frac{q^{(m+3)/2}-1}{q-1}$$ 
and dimension 
\begin{eqnarray}
q^m-2 \sum_{i=0}^{(m+1)/2} \binom{m}{i} (q-1)^i  + \frac{4+(q-2)(m+1)}{2}\binom{m}{\frac{m-1}{2}}. 
\end{eqnarray}
\end{theorem} 

\begin{proof}
The conclusion on the minimum distance comes from the BCH bound. We now prove the desired conclusion on 
the dimension of the code. It follows from Lemmas \ref{lem-RMc}, \ref{lem-april91} and \ref{lem-april92} that 
$a \in I{(q,m,(m+1)/2)} \cap (-I{(q,m,(m+1)/2)})$ if and only if one of the following three cases 
happens: 
\begin{enumerate}
\item[C1:] $\wt(a)=\frac{m-1}{2}$, $\wt(n-a)=\frac{m+1}{2}$, and the $q$-adic expression 
of $a$ is of the form 
$$ 
(q-1)(q^{i_1}+q^{i_2} + \cdots + q^{i_{(m-1)/2}}), 
$$
where $0 \leq i_1 < i_2 < \cdots < i_{(m-1)/2} \leq m-1$. The total number of such $a$'s is 
equal to $\binom{m}{(m-1)/2}$. 
\item[C2:] $\wt(a)=\frac{m+1}{2}$, $\wt(n-a)=\frac{m-1}{2}$, and the $q$-adic expression 
of $a$ is of the form 
$$ 
(q-1)(q^{i_1}+q^{i_2} + \cdots + q^{i_{(m+1)/2}}), 
$$
where $0 \leq i_1 < i_2 < \cdots < i_{(m+1)/2} \leq m-1$. The total number of such $a$'s is 
equal to $\binom{m}{(m+1)/2}$. 
\item[C3:] $\wt(a)=\frac{m+1}{2}$, $\wt(n-a)=\frac{m+1}{2}$, and the $q$-adic expression 
of $a$ is of the form 
$$ 
a_{i_1}q^{i_1}+a_{i_2}q^{i_2} + \cdots + a_{i_{(m+1)/2}}q^{i_{(m+1)/2}}, 
$$
where $0 \leq i_1 < i_2 < \cdots < i_{(m+1)/2} \leq m-1$, $1 \leq a_{i_{j}} \leq q-1$, 
and all the entries of the vector $(a_{i_1}, a_{i_2}, \cdots, a_{i_{(m+1)/2}})$ are $q-1$ 
except one that could be any element in $\{1,2, \cdots, q-2\}$. The total number of such $a$'s is 
equal to 
$$ 
\frac{(m+1)(q-2)}{2}\binom{m}{\frac{m+1}{2}}. 
$$
\end{enumerate} 

Summarizing the conclusions in the three cases above, we obtain that 
$$ 
|I{(q,m,(m+1)/2)} \cap (-I{(q,m,(m+1)/2)})|=\frac{4+(q-2)(m+1)}{2}\binom{m}{\frac{m-1}{2}}. 
$$

As before, let $g_{(q, m, (m+1)/2)}(x)$ be the polynomial of the code $\cGRM{(q,m, (m+1)/2)}$. 
Then the generator polynomial of $\overline{\cGRM}{(q, m, (m+1)/2)}$ is given by 
\begin{eqnarray*}
g(x) &=& (x-1)\lcm\left(g_{(q, m, (m+1)/2)}(x), \, g^*_{(q, m, (m+1)/2)}(x)\right) \\
&=& \frac{(x-1)g_{(q, m, (m+1)/2)}(x)g^*_{(q, m, (m+1)/2)}(x)}{\gcd\left(g_{(q, m, (m+1)/2)}(x), \,  g^*_{(q, m, (m+1)/2)}(x)\right)}.  
\end{eqnarray*}
Therefore, 
\begin{eqnarray*}
\deg(g(x)) 
&=& 2\deg(g_{(q, m, (m+1)/2)}(x))+1-\deg(\gcd(g_{(q, m, (m+1)/2)}(x), \, g^*_{(q, m, (m+1)/2)}(x)) \\
&=& 2\deg(g_{(q, m, (m+1)/2)}(x))+1- |I{(q,m,(m+1)/2)} \cap (-I{(q,m,(m+1)/2)})| \\ 
&=& 1+2 \sum_{i=1}^{(m+1)/2} \binom{m}{i} (q-1)^i  - \frac{4+(q-2)(m+1)}{2}\binom{m}{\frac{m-1}{2}}. 
\end{eqnarray*}
The desired conclusion on the dimension then follows. 
\end{proof} 

We point out that the dimension of $\overline{\cGRM}{(q, m, (m+1)/2)}$ is equal to zero 
when $q=2$. Hence, the code is nontrivial only when $q>2$. 

\begin{example} 
When $(q,m)=(4,3)$, $\overline{\cGRM}{(q, m, (m+1)/2)}$ has parameters $[63,8,42]$.  
\end{example} 

\begin{problem}\label{prob-five} 
Determine the minimum distance of the code $\overline{\cGRM}{(q, m, (m+1)/2)}$ of 
Theorem \ref{thm-bCode3}. 
\end{problem} 

\section{Summary and concluding remarks} 

The first contribution of this paper is the new generalisation of the classical punctured binary Reed-Muller codes. The newly generalised codes $\cGRM{(q, m, h)}$ are documented in Theorem \ref{thm-gRMcode}. A lower bound and a upper bound on the minimum distance of $\cGRM{(q, m, h)}$ were developed and given in Theorem \ref{thm-gRMcode}. Experimental data indicates that this lower bound is indeed the minimum distance. However, we were not able to prove this conjecture. It would be nice if this open problem can be settled. The dual code $\cGRM{(q, m, h)}^\perp$ was also studied. But the 
minimum distance $d^\perp$ of $\cGRM{(q, m, h)}^\perp$ is also open, though a lower bound on $d^\perp$ was given in Theorem \ref{thm-gRMcodedual}. The locality of both $\cGRM{(q, m, h)}$ and 
$\cGRM{(q, m, h)}^\perp$ depends on $d^\perp$ and $d$ respectively. Hence, it is also valuable to settle Open Problem \ref{prob-two}.     

The second contribution of this paper is to prove that the extended code $\widehat{\cGRM}(q,m,h)$ is 
affine-invariant and holds $2$-designs. 

The third contribution of this paper is the construction of the reversible cyclic codes 
$\overline{\cGRM}{(q, m, h)}$, which are based on the cyclic codes 
$\cGRM{(q, m, h)}$.  
The dimension of $\overline{\cGRM}{(q, m, h)}$ was settled for all $h$ with $1 \leq h \leq 
\lceil m/2 \rceil$. A lower bound on the minimum distance of the reversible cyclic code $\overline{\cGRM}{(q, m, h)}$ was developed. But the minimum distance of $\overline{\cGRM}{(q, m, h)}$ is unknown. It would be nice 
if Open Problems \ref{prob-three}, \ref{prob-four} and \ref{prob-five} can be resolved.

\end{document}